\documentclass[11pt,a4paper]{amsart}
\usepackage[margin=2.5cm]{geometry}

\usepackage[foot]{amsaddr}

\usepackage[utf8]{inputenc}
\usepackage{amsmath, amssymb, amsfonts,amsthm,amsopn}
\usepackage{dsfont}
\usepackage{graphicx}

\usepackage{titletoc}
\usepackage[section]{placeins}
\usepackage[shortlabels]{enumitem}
\setlist[enumerate]{itemsep=1em}
\usepackage{mathtools}

\usepackage{upgreek}
\usepackage{xcolor}
\usepackage{thm-restate}
\usepackage{tensor}

\usepackage[breaklinks]{hyperref}
\hypersetup{
    colorlinks,
    linkcolor={red!40!black},
    citecolor={blue!50!black},
    urlcolor={blue!20!black}
}
\usepackage{doi}
\usepackage{orcidlink}

\DeclareMathOperator{\tr}{tr}

\newcommand{\one}[0]{\mathds{1}}
\renewcommand{\a}{\alpha}
\renewcommand{\b}{\beta}

\newcommand{\R}{\mathds{R}}
\newcommand{\C}{\mathds{C}}

\newcommand{\Z}{\mathds{Z}}

\newcommand{\FFF}{\,\tensor[_2]{F}{_1}}

\DeclareMathOperator{\mc}{mc}
\DeclareMathOperator{\surplus}{sp}
\DeclareMathOperator{\qmc}{qmc}

\newtheorem{theorem}{Theorem}
\newtheorem*{theorem*}{Theorem}

\newtheorem{lemma}[theorem]{Lemma}
\newtheorem{definition}[theorem]{Definition}
\newtheorem{corollary}[theorem]{Corollary}

\newtheorem*{problem*}{Problem}
\newtheorem*{question*}{Question}

\newtheorem*{result*}{Result}

\newcommand{\nn}{\nonumber}
\newcommand{\overbar}[1]{\mkern 1.5mu\overline{\mkern-1.5mu#1\mkern-1.5mu}\mkern 1.5mu}

\begin{document}

\title
[Lovász theta and Shearer lower bounds on Quantum Max Cut]
{Lovász theta and Shearer lower bounds \\on Quantum Max Cut}

\date{\today}

\author{Felix Huber}
\address{
Felix Huber,
Division of Quantum Information,
Institute of Informatics,
Faculty of Mathematics, Physics and Informatics,
University of Gdańsk,
Wita Stwosza 57, 80-308 Gdańsk, Poland
}
\email{felix.huber@ug.edu.pl}

\thanks{
We thank Tomás Crosta and Ojas Parekh for fruitful discussions.
FH was funded in whole or in part by the National Science Centre, Poland 2024/54/E/ST2/00451
and by the Polish National Agency for Academic Exchange under the Strategic Partnership Programme grant BNI/PST/2023/1/00013/U/00001.
For the purpose of Open Access,
the author has applied a CC-BY public copyright licence to any Author Accepted Manuscript (AAM) version arising from this submission.
}

\begin{abstract}
Quantum Max Cut is a problem relevant to computer science and many-body quantum physics
due to its links to classical Max Cut and
the anti-ferromagnetic Heisenberg Hamiltonian.
We prove a lower bound to quantum Max Cut of a graph in terms of the Lovász theta function of its complement.
For a graph with $m$ edges,
 $\text{qmc}(G) \geq \tfrac{m}{4}\big(
 1 + \tfrac{8}{3\pi}\tfrac{1}{\vartheta(\overbar{G}) -1}
 \big)$, with the bound achieved by a product state.
The proof can be strenghtened by the vector chromatic number
and extends a result by Balla, Janzer, and Sudakov on classical Max Cut.
A relaxed bound follows from $\vartheta(\overbar{G}) - 1 \leq \Delta$ for graphs with maximum degree $\Delta$,
making it interesting for practically relevant quantum many-body systems.
We also extend results by Carlson et al. and
Shearer and show that $\qmc(G) \geq \frac{m}{4} + \frac{2m^{3/4}}{3 \pi}$ for all triangle-free graphs with $m$ edges.
\end{abstract}

\maketitle
\setcounter{tocdepth}{1}

\section{(Quantum) Max Cut}
Max Cut asks for a partition of the vertex set of a graph, such that the number of edges crossing the two sets is maximized.
In other words,
$\mc(G)$ is the size of the largest cut induced by a subgraph.
It can be modeled as the polynomial optimization problem,
\begin{align}\label{eq:mc}
 \max \quad & \frac{1}{2}\sum_{uv \in E} (1 - z_u z_v) \nn\\
 \text{subject to}\quad & z_u \in \{+1,-1\} \text{ for all } u \in V\,.
\end{align}

Let $\mc(G)$ be the size of the largest cut of $G$. Denote by $m =|E|$. The surplus is defined as
$\surplus(G) = \mc(G) - m/2$,
motivated by the fact that a value of $m/2$
can be obtained by a simple probabilistic method.
Balla, Janzer, and Sudakov~\cite{balla2023maxcutlovaszthetafunction} showed that
 $\surplus(G) \geq \frac{1}{\pi} \frac{m}{\vartheta(\overbar{G}) -1}$,
where $\vartheta(\overbar{G})$ is the Lovász theta function~\cite{1055985} of the complement graph $\overbar{G}$.
Equivalently,
\begin{equation}\label{eq:bound_mc}
 \mc(G) \geq \frac{m}{2}\Big(1 +
 \frac{2}{\pi} \frac{1}{\vartheta(\overbar{G}) -1}\Big)\,.
\end{equation}

The result originates in the following formulation of the Lovász theta function.
\begin{definition}[\cite{10.1145/274787.274791}]\label{def:Lovasz}
The Lovász theta function of a graph $G$, denoted $\vartheta(G)$, is the minimal $\kappa \geq 2$ for which there exists a unit vector $x_v$ (in some Euclidean space) for every vertex $v$ such that $\langle x_u, x_v \rangle = - \tfrac{1}{\kappa-1}$ holds whenever $u$ and $v$ are distinct vertices in $G$ and $uv \not \in E(G)$.
\end{definition}
Without restriction of generality, the vectors $x_u$ can be chosen to lie in $\R^n$,
as $n$ vectors span an at most $n$-dimensional space.
The definition of $\vartheta(\overbar{G})$ has
$uv \not \in E(G)$ replaced by $uv \in E(G)$.

{\it Quantum Max Cut} (QMC) is a
generalization of Max Cut [Eq.~\eqref{eq:mc}] to operators,
and asks for the largest eigenvalue of the Hamiltonian
\begin{equation}\label{eq:qmc}
 H^{\qmc} = \frac{1}{4} \sum_{uv \in E} (I - X_u X_v - Y_u Y_v - Z_u Z_v) \,.
\end{equation}
Here $X_u,Y_u,Z_u$ are the tensor product operators acting with the Pauli matrices $X,Y,Z$ on vertex $u$ and with identity on the remaining vertices.
Denote by $\qmc(G)$ the {largest} eigenvalue of $H^{\qmc}$.
The study of $H^{\qmc}$ is also physically motivated. Quantum magnetic systems can be modeled by the anti-ferromagnetic Heisenberg Hamiltonian $H^{\text{H}} = - H^{\qmc}$,
and one is interested in the ground state energy, i.e. its smallest eigenvalue.
Naturally, the task of finding the smallest eigenvalue of
$H^{\text{H}}$ is equivalent to finding the largest eigenvalue of $H^{\qmc}$.
We restrict ourselves to the maximization problem $H^{\qmc}$.
A range of approximation algorithms have been designed for $H^{\text{qmc}}$
that use the moment-SOS hierarchy~
\cite{
parekh2022optimal,
King2023improved,
lee2024improved,
huber2024secondorderconerelaxations,
gribling2025improvedapproximationratiosquantum,
apte2025improvedalgorithmsquantummaxcut}.
A trivial lower bound of ${m}/{4}$ on $\qmc$ is obtained by the maximally mixed state.

\subsection{Notation}
We will use both $\langle u, v \rangle$ and $u\cdot v$
for the scalar product of two vectors.
The identity matrix is $\one$ and the three Pauli matrices are
\begin{equation}
 \sigma_x = \begin{pmatrix}
             0 & 1 \\ 1 & 0
            \end{pmatrix}\,,\quad
 \sigma_y = \begin{pmatrix}
             0 & -i \\ i & 0
            \end{pmatrix}\,,\quad
\sigma_z =  \begin{pmatrix}
             1 & 0 \\ 0 & -1
            \end{pmatrix}\,.
\end{equation}
To avoid a proliferation of indices,
we sometimes write
$X$ for $\sigma_x$, $Y$ for $\sigma_y$, and $Z$ for $\sigma_z$.
The operator $P_u$ denotes the tensor product operator
acting with $P$ on vertex $u$ and with identity $\one$ on the remaining vertices.

\section{Briët-de Oliveira Filho-Vallentin rounding}
A key ingredient for the Lovász and Shearer product state lower bounds is the following:

\begin{lemma}[\cite{v010a004}, Lemma  2.1]
\label{lem:BdOFV}
Let $u, v$ be unit vectors in $\R^n$ and let
$Z \in \R^{r \times n}$
be a random matrix whose entries are distributed independently according to the standard normal
distribution with mean $0$ and variance $1$. Then,
\begin{equation}
\label{eq:BdOFV}
 \mathds{E}\Big[\frac{Zu}{||Zu||} \cdot \frac{Zv}{||Zv||}\Big]
 = \frac{2}{r}\Bigg(
 \frac{\Gamma((r+1)/2)}
 {\Gamma(r/2)}
 \Bigg)^2
\hat F(r,u  \cdot v)\,.
 \end{equation}
Here
\begin{align}\label{eq:BdOFV_F}
\hat F(r,u  \cdot v) =
 (u  \cdot v)\,\  &\FFF
 \Bigg( {1/2 ,1/2 \atop  r/2+1}
 ;(u\cdot v)^2
 \Bigg) \nn\\
 = \,\,&\sum_{k=0}^\infty
 \frac{
 (1\cdot 3 \cdots (2k-1))^2}
 {(2\cdot 4 \cdots 2k)
  ((r+2)\cdot (r+4) \cdots (r+2k))
  }
  (u\cdot v)^{2k+1}\,,
\end{align}
with \,$\FFF$ the hypergeometric function.
\end{lemma}
Setting $r=1$
and
 $\arcsin(x) = x\! \cdot \!\FFF(\tfrac{1}{2},\tfrac{1}{2},\tfrac{3}{2}; x^2)$
recovers Grothendieck's identity~\cite{AIF_1952__4__73_0},
\begin{equation}\label{eq:Grothendieck}
\mathds{E}\big[\operatorname{sign}(Zu) \cdot \operatorname{sign}(Zv)\big]
= \frac{2}{\pi} \arcsin(u\cdot v)\,.
\end{equation}

Let $G$ be a graph and $\overbar{G}$ its complement.
By Definition~\ref{def:Lovasz}
there exists a vector $x_u$ for every vertex $u$
such that
$\langle x_u, x_v \rangle = - 1/(\vartheta(\overbar{G}) -1)$
on every edge $uv\in E(G)$.
This motivates the following.
\begin{lemma}\label{lem:bound_F}
If $\langle x_u, x_v\rangle \leq 0$, then
\begin{equation}
 \hat F(r, \langle x_u ,x_v\rangle) \leq \langle x_u , x_v\rangle \,.
\end{equation}
\end{lemma}
\begin{proof}
 Inspect the Taylor series in Eq.~\eqref{eq:BdOFV_F}.
 For the given parameters,
 each coefficient in $\FFF$ is strictly positive with the zero'th coefficient equal to one.
 Thus $\FFF \geq 1$.
 Furthermore, every power in $\hat F(r, \langle u, v \rangle)$ is odd, from which the claim follows.
\end{proof}

\section{A Lovász theta lower bound}

\begin{theorem}\label{prop:lower_bound}
Let $G$ be a graph with $m$ edges.
Then
\begin{equation}\label{eq:lower_bound_qmc}
 \qmc(G)
 \geq
 \frac{m}{4}\bigg(
 1 + \frac{8}{3\pi}\frac{1}{\vartheta(\overbar{G}) -1}
 \bigg)\,.
\end{equation}
The bound can be obtained with product states.
\end{theorem}
\begin{proof}
For a given graph $G$, let $\{x_u \in \R^n\,:\, u \in V(G)\}$
be vectors that realize the Lovász theta function
$\vartheta(\overbar{G})$
of its complement $\overbar{G}$.
Then on every edge $uv \in E(G)$,
\begin{equation}\label{eq:Lovasz_vector}
\langle x_u, x_v\rangle = - \frac{1}{\vartheta(\overbar{G})-1} \,.
\end{equation}

Following the product state rounding by Gharibian and Parekh~\cite{gharibian_et_al:LIPIcs:2019:11246},
choose a random matrix $Z \in \R^{3 \times n}$
whose entries are distributed independently according to the standard normal distribution with mean $0$ and variance $1$. Let
\begin{equation}\label{eq:round}
 y_u = \frac{Zx_u}{||Zx_u||}\,.
\end{equation}
For each vertex $u$ define the one-qubit state
\begin{equation}
 \varrho_u = \frac{1}{2} \big(\one + y_u^1 \sigma_x + y_u^2 \sigma_y + y_u^3 \sigma_z\big) \,.
\end{equation}
where $y_u^1,y_u^2,y_u^3$ are the three components of $y_u$.
Note that
$\tr(\varrho)=1$
and
$\varrho \succeq 0$
due to
$\langle y_u, y_u \rangle = 1$,
so that the coefficients of $\sigma_x, \sigma_y, \sigma_z$
lie on the Bloch sphere.
Let
\begin{equation}
 \varrho^{\vartheta} = \prod_{u \in V(G)} \varrho_u\,.
\end{equation}
Then
\begin{align}\label{eq:ineq}
 \mathds{E}\big[\tr(H^{{\qmc}}\varrho^{\vartheta})\big]
 &=
 \frac{1}{4}\sum_{uv \in E} \big(
 1 - \mathds{E}[y_u \cdot y_v]
 \big) \nn\\
 &=  \frac{1}{4}\sum_{uv \in E} \Big(
 1 -
 \frac{2}{r}\Bigg(
 \frac{\Gamma((r+1)/2)}
 {\Gamma(r/2)}
 \Bigg)^2
 \hat F(r, \langle x_u, x_v\rangle)
 \Big)
 &
 \text{[Lemma~\ref{lem:BdOFV}]}
 \nn\\
 &\geq\frac{1}{4}
 \sum_{uv \in E} \Big(
 1 -
 \frac{8}{3\pi} \langle x_u ,x_v \rangle
 \Big)
 &
 \text{[Lemma~\ref{lem:bound_F}; Eq.~\eqref{eq:Gamma_value}]}\nn\\
 &=\frac{1}{4}
 \sum_{uv \in E} \Big(
 1 +
 \frac{8}{3\pi} \frac{1}{\vartheta(\overbar{G}) - 1)}
 \Big)\,.
  &
 \text{[Definition~\ref{def:Lovasz}]}
\end{align}
Above, we have used Definition~\ref{def:Lovasz},
Lemma \ref{lem:BdOFV} and Eq.~\ref{lem:bound_F} with $r=3$,
and the below Eq.~\eqref{eq:Gamma_value},
\begin{equation}\label{eq:Gamma_value}
 \frac{2}{3}\Bigg(
 \frac{\Gamma((3+1)/2)}
 {\Gamma(3/2)}
 \Bigg)^2
 = \frac{8}{3\pi}\,.
\end{equation}
Eq.~\eqref{eq:ineq} is the expected value of
$\tr(H^{{\qmc}}\varrho^{\vartheta})$
after having obtained $\varrho^{\vartheta}$ by rounding with a random matrix $Z$.
Thus there exists at least one state achieving the bound.
It follows that
\begin{equation}
 \qmc(G)
 \geq
 \mathds{E}\big[\tr(H^{{\qmc}}\varrho^{\vartheta})\big]
 \geq
 \frac{m}{4}\bigg(
 1 + \frac{8}{3\pi}\frac{1}{\vartheta(\overbar{G}) - 1}
 \bigg)\,.
\end{equation}
This ends the proof.
\end{proof}
The classical bound by Balla, Janzer, and Sudakov is obtained with setting $r=1$, and the fact that $\arcsin(x) \leq x$ for all $x\leq 0$,
which is a special case of Lemma~\ref{lem:bound_F}.
The rounding to quantum states outperforms the rounding to classical states when applied to $H^{\qmc}$
due to $0.8488 \approx \tfrac{8}{3\pi} > \tfrac{2}{\pi}\approx 0.6366$.
Ref.~\cite{balla2023maxcutlovaszthetafunction}
also showed that $\vartheta(\overbar{G}) \leq \Delta + 1$,
where $\Delta$ is the maximum degree of the graph.
The relaxed bound
$\text{qmc}(G) \geq \tfrac{m}{4}\big(
 1 + \tfrac{8}{3\pi}\tfrac{1}{\Delta}
 \big)$
is interesting for bounded-degree graphs that are relevant in quantum physics,
e.g. lattices with defects.

Similar to the classical bound, a slightly stronger result on QMC
can be made by replacing $\vartheta(\overbar{G})$ in
 Eq.~\eqref{eq:lower_bound_qmc}
by the vector chromatic number $\chi_\text{vec}(G)$.
One uses the fact that
$
 \chi_\text{vec}(G) \leq \vartheta(\overbar{G}) \,,
$
which follows from:
\begin{definition}[\cite{10.1145/274787.274791}]
The vector chromatic number of a graph $G$, denoted $\chi_\text{vec}(G)$, is the minimal $\kappa \geq 2$ for which there exists a unit vector $x_v$ (in some Euclidean space) for every vertex $v$ such that $\langle x_u, x_v \rangle \leq - \tfrac{1}{\kappa-1}$ holds whenever $u$ and $v$ are distinct vertices in $G$ and $uv \in E(G)$.
\end{definition}
Note the inequality $\leq$ compared to equality $=$ in Definition~\ref{def:Lovasz} for the Lovász theta function.

The method naturally applies also to Hamiltonians that contain only some of the QMC terms, e.g.,
$ H^{\text{XX}} = \tfrac{1}{4} \sum_{uv \in E} (I - X_u X_v - Y_u Y_v)$.
Then $r=2$ in $\hat F(r, u  \cdot v)$
and
\begin{equation}\label{eq:lower_bound_XX}
 \qmc^{XX}(G)
 \geq
 \frac{m}{4}\bigg(
 1 + \frac{\pi}{4}\frac{1}{\vartheta(\overbar{G}) -1}
 \bigg)\,.
\end{equation}

\section{A Shearer lower bound}
The same rounding strategy applies to other approximations that rely on concrete
Ansatzes for the SDP relaxation of classical Max Cut,
such as found in Refs.~\cite{
doi:10.1137/1.9781611977066.17, carlson2020lowerboundsmaxcuthfree}.
The idea is that
\begin{align}\label{eq:mc_relax}
 \max \quad & \frac{1}{2} \sum_{(uv) \in E} (1-Z_{uv})  \nn\\
 \text{subject to} \quad & Z_{uu} = 1  \quad \text{ for all }u \in V \,, \nn\\
 & Z \succeq 0\,,
\end{align}
relaxes Eq.~\eqref {eq:mc}.
Any feasible point to the SDP~\eqref{eq:mc_relax}
can then be rounded to a classical solution via a Cholesky decomposition
and rounding of the resulting vectors.
While in the previous Section we inferred a feasible set of vectors from the Lovász number,
{\it any} feasible point of the SDP~\eqref{eq:mc_relax} can be rounded.

For example, let us consider
Ref.~\cite{carlson2020lowerboundsmaxcuthfree}.
Its Theorem $1.1$
implies Shearer's bound on classical Max Cut, which states that for all triangle-free graphs~\cite[Corollary 1]{https://doi.org/10.1002/rsa.3240030211},
\begin{equation}\label{eq:shearer_classical}
 \mc(G) \geq \frac{m}{2} + \frac{1}{8\sqrt{2}} \sum_{v\in V} \sqrt{d_v}\,.
\end{equation}
Here $d_v$ are the degrees of the vertices $v$. From Eq.~\eqref{eq:shearer_classical}
follows that for all triangle-free graphs~\cite[Corollary 2]{https://doi.org/10.1002/rsa.3240030211},
\begin{equation}\label{eq:shearer_classical2}
 \mc(G) \geq \frac{m}{2} + \frac{1}{8\sqrt{2}} m^{3/4}\,.
\end{equation}

To obtain Shearer-type bounds for quantum Max Cut, we need to extend the following inequalities that hold for
$\arcsin$ to the function $\hat F$ from Eq.~\eqref{eq:BdOFV_F}:
\begin{alignat}{2}
 \arcsin(\a) &\leq \arcsin(\b) \quad &&\text{for } \a \leq \b\,, \nn\\
 \arcsin(x) &\leq x                      &&\text{for } x \in [-1,0]\,,\nn\\
 \arcsin(x) &\leq \tfrac{\pi}{2} x \quad &&\text{for } x \in [0,1]\,,
 \nn\\
 \arcsin(\a-\b) &\leq \tfrac{\pi}{2}\a - \b
 &&\text{for } \a,\b \in [0,1]\,.
\end{alignat}
We will state them in terms of inner products $\langle x_u, x_v\rangle$ of unit vectors whenever appropriate.

\begin{lemma}\label{lem:F}
The following relations hold.
\begin{enumerate}[label=(\roman*)]
 \item
 Let $\alpha\leq \beta$. Then $\hat F(r,\alpha) \leq \hat F(r,\beta)$.

 \item
If $\langle x_u, x_v\rangle \leq 0$, then
\begin{equation}
 \hat F(r, \langle x_u ,x_v\rangle) \leq \langle x_u , x_v\rangle \,.
\end{equation}

 \item If $\langle x_u, x_v \rangle \geq 0$, then
\begin{equation}
 \hat F(r, \langle x_u ,x_v\rangle) \leq
 \frac{\Gamma(c)\Gamma(c-a-b)}{\Gamma(c-a) \Gamma(c-b)}
 \,\langle x_u ,x_v\rangle\,,
\end{equation}
where $a=b=\tfrac{1}{2}$ and $c = \tfrac{r}{2}+1$.

\item
Let $\a,\b \in [0,1]$ be non-negative. Then,
 \begin{equation}
 \hat F(3, \a-\b) \leq \frac{3\pi}{8}\a - \b\,.
 \end{equation}

\end{enumerate}
\end{lemma}

\begin{proof}
Let us show items (i) - (iv).
\begin{enumerate}[label=(\roman*)]

 \item
 This is seen by term-wise comparison of the Taylor series [Eq.~\eqref{eq:BdOFV_F}].

 \item This was shown in Lemma~\ref{lem:bound_F}.

 \item
 For $\operatorname{Re}(c-a-b) > 0$, the Gauss summation formula holds,
 \begin{equation}
 \FFF( {a ,b, c} ; 1) = \frac{\Gamma(c)\Gamma(c-a-b)}{\Gamma(c-a) \Gamma(c-b)}\,.
\end{equation}
Then note that for all $t \in [0,1]$,
\begin{equation}
\FFF( {a ,b, c} ; t ) \leq \FFF( {a ,b, c} ; 1)\,,
\end{equation}
which can be seen from the Taylor series in Eq.~\eqref{eq:2F1_taylor}.
As a consequence,
\begin{align}
 \hat F(r, \langle x_u ,x_v\rangle)
 & \leq \frac{\Gamma(c)\Gamma(c-a-b)}{\Gamma(c-a) \Gamma(c-b)}
 \, \langle x_u ,x_v\rangle\,.
\end{align}

\item
 The first case $\a-\b \leq 0$ follows from Item (ii) (respectively Lemma~\ref{lem:bound_F}),
 \begin{equation}
  \hat F(3, \a-\b) \leq (\a-\b) \leq \frac{3\pi}{8}\a - \b\,.
 \end{equation}
 The second case $\a-\b \geq 0$ follows from Item (iii),
 \begin{equation}
  \hat F(3, \a-\b) \leq \frac{3\pi}{8}(\a-\b) \leq \frac{3\pi}{8}\a - \b\,,
 \end{equation}
 where we used that for $a=b=1/2$ and $c=3/2+1$,
\begin{equation}
 \frac{\Gamma(c)\Gamma(c-a-b)}{\Gamma(c-a) \Gamma(c-b)} = \frac{3\pi}{8}\,.
\end{equation}

\end{enumerate}
This ends the proof.
\end{proof}

Note that the above relations reduce to those for $\arcsin$ in the interval $[-1,1]$.
In particular,
for $r=1$, Item (i) reduces to $\arcsin(x) \leq \tfrac{\pi}{2}x$ for $x\in [0,1]$.
For $r=3$, it reads
\begin{equation}
 \hat F(3, \langle x_u ,x_v\rangle) \leq \frac{3\pi}{8} \langle x_u ,x_v\rangle\,.
\end{equation}

We can now show the bound corresponding to that of
Ref.~\cite[Theorem 1.1]{
carlson2020lowerboundsmaxcuthfree}:
\begin{theorem}\label{thm:carlson_quantum}
 For every $v\in V$, let $V_v$ be any subset of neighbors of vertex $v$ and $\epsilon_v \leq \frac{1}{\sqrt{|V_v|}}$. Then,
\begin{equation}\label{eq:carlson}
 \qmc(G) \geq \frac{m}{4}
 + \sum_{v \in V} \frac{2\epsilon_v|V_v|}{3\pi}
 - \sum_{vw \in E} \frac{\epsilon_v \epsilon_w|V_v \cap V_w| }{4}\,,
\end{equation}
with the bound obtained by a product state.
\end{theorem}
\begin{proof}

We follow Ref.~\cite{
carlson2020lowerboundsmaxcuthfree}.
Construct a set of vectors $\tilde x_v$ with entries indexed by~$w$ as
\begin{equation}
 (\tilde  x_v)_w = \begin{cases}
        1 \quad \quad \text{if } v=w
        \\
        - \epsilon_i \quad \text{if } w \in V_v
        \\
        0 \quad\quad \text{otherwise}
       \end{cases}
\end{equation}
and normalize $x_v = \tilde x_v / ||\tilde x_v||$.
Then~\cite{carlson2020lowerboundsmaxcuthfree},
\begin{equation}
 \langle x_v,  x_w \rangle \leq
 -\frac{\epsilon_v}{4}\one_{V_w}(v)
 -\frac{\epsilon_w}{4}\one_{V_v}(w)
 + |V_v \cap V_w| \epsilon_v \epsilon_w\,.
\end{equation}
where $\one_{V_w}(v) = 1$ if $v \in V_w$ and $0$ otherwise.

Now perform a rounding as in Theorem~\ref{prop:lower_bound},
to obtain $y_u = Zx_u / ||Zx_u||$
and a corresponding product state $\varrho^C$.
We use the derived relations on $\hat F$ to bound its expectation value.
\begin{align}
 \mathds{E}\big[\tr(H^{{\qmc}}\varrho^{C})\big]
 &=
 \frac{1}{4}\sum_{vw \in E} \big(
 1 - \mathds{E}[y_v \cdot y_w]
 \big) \nn\\
 &=  \frac{1}{4}\sum_{uv \in E} \Big(
 1 -
 \frac{2}{3}\Bigg(
 \frac{\Gamma((3+1)/2)}
 {\Gamma(3/2)}
 \Bigg)^2
 \hat F(3, \langle x_u, x_v\rangle)
 \Big)
 &
 \text{[Lemma~\ref{lem:BdOFV}]}
 \nn\\
 &\geq\frac{1}{4}
 \sum_{uv \in E} \big(
 1 - \frac{8}{3\pi}
 \hat F\Big(3,
 |V_v \cap V_w| \epsilon_v \epsilon_w
 -\frac{\epsilon_v}{4}\one_{V_w}(v)
 -\frac{\epsilon_w}{4}\one_{V_v}(w)
 \Big)
 &
 \text{[Lemma~\ref{lem:F}, (i);  Eq.~\eqref{eq:Gamma_value}]}\nn\\
 &\geq\frac{1}{4}
 \sum_{uv \in E} \bigg(
 1 - \frac{8}{3\pi} \bigg[
 \frac{3\pi}{8}
 |V_v \cap V_w| \epsilon_v \epsilon_w
 - \frac{\epsilon_v}{4}\one_{V_w}(v)
 - \frac{\epsilon_w}{4}\one_{V_v}(w) \bigg]
 \bigg)
 &\text{[Lemma~\ref{lem:F}, (iv)]}\nn\\
 &= \frac{m}{4}
 + \sum_{v \in V} \frac{2\epsilon_v|V_v|}{3\pi}
 - \sum_{vw \in E} \frac{\epsilon_v \epsilon_w|V_v \cap V_w| }{4}\,.
\end{align}
This ends the proof.
\end{proof}

In turn, this implies a quantum analogue of
Shearer's bound~\cite{https://doi.org/10.1002/rsa.3240030211} in Eq.~\eqref{eq:shearer_classical}:
\begin{corollary}\label{cor:qshearer1}
 Let $G$ be a triangle-free graph
 with vertex degrees $d_v$. Then,
\begin{equation}
 \qmc(G) \geq \frac{m}{4} + \frac{2}{3\pi}\sum_{v\in V} \sqrt{d_v}  \,,
\end{equation}
with the bound obtained by a product state.
\end{corollary}
\begin{proof}
 In Theorem~\ref{thm:carlson_quantum} set $V_v$ to the neighbors of $v$ and $\epsilon_v = \tfrac{1}{\sqrt{d_v}}$ for all
 $v \in V$. Since $G$ is a triangle-free graph, $|V_v \cap V_w| = 0$ for all $(vw) \in E$. This ends the proof.
\end{proof}
Ref.~\cite[Lemma~2]{https://doi.org/10.1002/rsa.3240030211}
states that for all graphs with $m$ edges,
\begin{equation}
  \sum_{v\in V} \sqrt{d_v} \geq m^{3/4}\,.
\end{equation}
Thus Corollary~\ref{cor:qshearer1}
can be relaxed to:
\begin{corollary}\label{cor:qshearer2}
Let $G$ be a triangle-free graph. Then,
 \begin{equation}
 \qmc(G) \geq \frac{1}{4}m + \frac{2}{3 \pi}m^{3/4}\,.
\end{equation}
with the bound obtained by a product state.
\end{corollary}
Corollary~\ref{cor:qshearer2} can be seen as the quantum analogue
of Shearer's bound on classical Max Cut
in Eq.~\eqref{eq:shearer_classical2}~\cite[Corollary 2]{https://doi.org/10.1002/rsa.3240030211}.

\appendix
\section{Product state rounding}
The proof of Theorem~\ref{prop:lower_bound} is inspired by
the classical result of Balla, Janzer, and Sudakov and
by the randomized rounding method of
Briët, de Oliveira Filho, and Vallentin,
which was used by Gharibian and Parekh to obtain
a product state approximation ratio for quantum Max Cut~\cite{gharibian_et_al:LIPIcs:2019:11246}.
Balla, Janzer, and Sudakov corresponds to the case of $r=1$,
which is an instance of the Goemans-Williamson randomized
rounding strategy~\cite{10.1145/227683.227684}.
We sketch the key idea for obtaining quantum product states,
which uses $r=3$.

Index a $3n\times 3n$ real moment matrix $M$
by the Cartesian product of
qubit position $i \in \{1, \dots, n\}$ and choice of
Pauli matrix $k\in\{X,Y,Z\}$. E.g. the index $iX$ corresponds to the Pauli matrix $X$ acting on qubit $i$, that is $\sigma_x^{(i)}$.
Given a state $\varrho$,
form a moment matrix with entries
\begin{equation}
M(ik,jl) = \tr\big(\varrho \, \sigma_k^{(i)} \sigma_l^{(j)}\big) \,.
\end{equation}
Then $M(ik,jl) = M(jl,ik)$ for all $i\neq j$ as Pauli matrices acting on different subsystems commute.
But on the same coordinate two different Pauli matrices anti-commute,
and so
$M(ik,il) = - M(il,ik)$.
By construction, $M\succeq 0$, and the ground state energy
$\tr(H^{\qmc}\varrho)$ can be extracted from $M$ as a linear combination of its entries.

The idea is then to relax the optimization over states $\varrho$
to an optimization over all positive-semidefinite matrices satisfying the above constraints.
The objective function is invariant under a global transpose of $M$ and any imaginary entries (corresponding to anti-commuting pairs of Pauli matrices) can be set to zero.

A Cholesky decomposition of the optimal $M$ results in $3n$ unit vectors $v_{ik} \in \R^{3n}$ so that
$M(ik,jl) = v_{ik}^T v_{jl}$.
Now stack the vectors belonging to a single vertex,
\begin{equation}
u_i =  \begin{pmatrix}
  v_{iX} \\ v_{iY} \\ v_{iZ}
 \end{pmatrix},
\end{equation}
and normalize,
$
x_i = u_i \,/\, ||u_i||\,.
$
Round to vectors in $\R^3$ with a random matrix $R \in \R^{3 \times 3n}$,
where $R$ is a random matrix whose entries are distributed independently according to the standard normal distribution with mean $0$ and variance $1$,
$
y_i = {Rx_i}/{||Rx_i||}\,.
$
Define the product state $\varrho^{\text{GP}} = \prod_{i=1}^n \varrho_i$, where
$\varrho_i = \frac{1}{2} \big(\one + y_i^x \sigma_x + y_i^y \sigma_y + y_i^z \sigma_z\big)$.
Using the rounding scheme by Briët, de Oliveira Filho, and Vallentin [Lemma~\ref{lem:BdOFV}],
Gharibian and Parekh \cite{gharibian_et_al:LIPIcs:2019:11246}
showed that this randomized rounding approach
achieves an approximation ratio of
\begin{equation}\label{eq:taylor}
 \frac{\mathds{E}[\tr(H^{\qmc} \varrho^{\text{GP}})]}{\qmc(G)} \geq 0.498\,.
\end{equation}

\section{Hypergeometric function}
The function $\FFF$ is for $|z|<1$ defined by
\begin{equation}\label{eq:2F1_taylor}
 \FFF
 ( {a ,b, c} ; z)
 =
 \sum_{j=0}^\infty \frac{(a)_j (b)_j}{(c)_j} \frac{z^j}{j!}\,.
\end{equation}
where the Pochhammer symbol or raising factorial is for $a\in \C$ and $n \in \Z$ defined as
\begin{align}
(a)_0 &= 1\nn\\
(a)_n &= a(a+1)\dots (a+n-1) &&n=1,2\dots \nn\\
(a)_n &= \frac{1}{(a-n) \dots (a-1)} &&n=\dots, -2,-1\,.
\end{align}
It is clear that the zero'th term of $\FFF$ evaluates to $1$.

For $\operatorname{Re}(c-a-b) > 0$
there exits an analytic continuation at $z=1$.
The Gauss summation formula states that,
\begin{equation}
 \FFF( {a ,b, c} ; 1 ) = \frac{\Gamma(c)\Gamma(c-a-b)}{\Gamma(c-a) \Gamma(c-b)}\,.
\end{equation}

\bibliographystyle{amsalpha}
\bibliography{current_bib}

\end{document}